\documentclass[11pt]{article}

\usepackage[margin=1in]{geometry}
\usepackage[T1]{fontenc}
\usepackage{lmodern}
\usepackage{microtype}
\usepackage{amsmath,amssymb,amsthm,mathtools}
\usepackage{enumitem}
\usepackage{aliascnt}
\usepackage{booktabs,tabularx,array}
\usepackage{xcolor}
\usepackage{tikz}
\usetikzlibrary{arrows.meta,positioning}
\usepackage[hidelinks]{hyperref}
\usepackage[nameinlink,noabbrev]{cleveref}

\newtheorem{theorem}{Theorem}[section]
\newaliascnt{lemma}{theorem}
\newtheorem{lemma}[lemma]{Lemma}
\aliascntresetthe{lemma}
\newaliascnt{proposition}{theorem}
\newtheorem{proposition}[proposition]{Proposition}
\aliascntresetthe{proposition}
\newaliascnt{corollary}{theorem}
\newtheorem{corollary}[corollary]{Corollary}
\aliascntresetthe{corollary}
\theoremstyle{definition}
\newaliascnt{definition}{theorem}

\aliascntresetthe{definition}
\newaliascnt{problem}{theorem}
\newtheorem{problem}[problem]{Problem}
\aliascntresetthe{problem}
\theoremstyle{remark}
\newaliascnt{remark}{theorem}

\aliascntresetthe{remark}
\newtheorem*{informaltheorem}{Main theorem (informal)}

\crefname{theorem}{Theorem}{Theorems}
\crefname{lemma}{Lemma}{Lemmas}
\crefname{proposition}{Proposition}{Propositions}
\crefname{corollary}{Corollary}{Corollaries}
\crefname{definition}{Definition}{Definitions}
\crefname{problem}{Problem}{Problems}
\crefname{remark}{Remark}{Remarks}
\Crefname{theorem}{Theorem}{Theorems}
\Crefname{lemma}{Lemma}{Lemmas}
\Crefname{proposition}{Proposition}{Propositions}
\Crefname{corollary}{Corollary}{Corollaries}
\Crefname{definition}{Definition}{Definitions}
\Crefname{problem}{Problem}{Problems}
\Crefname{remark}{Remark}{Remarks}

\newcommand{\cl}{\operatorname{cl}}
\newcommand{\Int}{\operatorname{Int}}
\newcommand{\MI}{\operatorname{MI}}

\newcommand{\Anc}{\operatorname{Anc}}
\newcommand{\Agr}{\operatorname{Agr}}
\newcommand{\C}{\mathcal C}
\newcommand{\M}{\mathcal M}
\newcommand{\B}{\mathcal B}
\newcommand{\G}{\mathcal G}
\newcommand{\K}{\mathbb K}
\newcommand{\set}[1]{\left\{#1\right\}}
\newcommand{\suchthat}{\,\middle|\,}
\newcommand{\problemname}[1]{\textnormal{\textsc{#1}}}
\newcolumntype{Y}{>{\raggedright\arraybackslash}X}

\setlist[itemize]{leftmargin=1.5em,itemsep=0.3em,topsep=0.4em}
\setlist[enumerate]{leftmargin=1.7em,itemsep=0.3em,topsep=0.4em}

\hypersetup{
  pdftitle={Completeness of Canonical Closure Representations Is coNP-Complete},
  pdfauthor={Mikhail Babin},
  pdfsubject={Closure systems, Horn logic, enumeration complexity, formal concept analysis, and functional dependencies},
  pdfkeywords={closure systems, implicational bases, meet-irreducibles, characteristic models, pseudo-intents, Duquenne-Guigues basis, functional dependencies, coNP-completeness}
}

\title{Completeness of Canonical Closure Representations Is coNP-Complete}
\author{Mikhail Babin\\[0.25em]
\normalsize\href{mailto:bmisterxster@gmail.com}{\texttt{bmisterxster@gmail.com}}}
\date{17 July 2026}

\begin{document}
\maketitle

\begin{abstract}
Across Horn logic, Formal Concept Analysis (FCA), convex geometries, and database theory, the same basic representation problem has remained open for about thirty years. A finite closure system on a finite set $U$ is a family of subsets that contains $U$ and is closed under intersections. It can be specified in two elementary ways. An \emph{implicational specification} lists rules $A\to b$ and consists of all subsets $X\subseteq U$ that satisfy every rule: if $A\subseteq X$, then $b\in X$. An \emph{intersection specification} lists subsets $M_1,\ldots,M_t$ and consists of all intersections of subfamilies of that list. We ask whether one specification of each kind defines exactly the same family of subsets.

In 1995, Khardon showed that efficiently translating between Horn formulas and their characteristic models is equivalent to deciding whether a proposed list of characteristic models is complete, but left the exact complexity open. Three decades later, at ISAAC 2025, the corresponding problem of enumerating irreducible closed sets from implications was still described as ``widely open,'' even for acyclic convex geometries. Closely connected open questions concerned pseudo-intents and the Duquenne--Guigues basis in FCA, and functional dependencies and Armstrong relations in database theory.

We prove that the equivalence test is coNP-complete. Hardness holds for acyclic implications whose premises contain at most three elements, even when every listed subset is correct and no listed subset can be removed without changing the closure system generated by the list. Thus the hard part is simply deciding whether a required set is missing. Unless $\mathrm P=\mathrm{NP}$, the complete canonical lists cannot in general be generated in time polynomial in the input plus the total output size, even for acyclic convex geometries. Through standard correspondences, the theorem makes Characteristic Models Identification and FD--Relation Equivalence coNP-complete and rules out such output-polynomial algorithms for Horn characteristic models, all pseudo-intents (equivalently, the Duquenne--Guigues basis), and premises of minimum functional-dependency covers.
\end{abstract}


\paragraph{AI-assistance disclosure.}
The main proof was obtained by OpenAI's GPT-5.6 Pro through ChatGPT and checked by the author. The author takes full responsibility for the mathematical claims and the final manuscript.

\section{Introduction}

\subsection{A long-standing problem with several names}

A recurring task in symbolic computation is to move between two descriptions of the same finite structure. One description is a compact collection of rules. The other is a canonical list of extremal models or counterexamples. Rules support fast deduction; the canonical list supports model-based reasoning, reduced data representations, and translation into other formalisms. Because the canonical list may be exponentially larger than the rules, the natural target is output-sensitive:
\begin{quote}
Can the complete canonical list be generated in time polynomial in the input plus the output, and can a proposed list be certified as complete?
\end{quote}

Variants of this question have been studied for more than thirty years in communities that often use different terminology. A Dagstuhl seminar devoted to Horn formulas, directed hypergraphs, lattices, closure systems, database dependencies, and Formal Concept Analysis emphasized that closely related problems had long been studied in parallel~\cite{AdarichevaEtAl2014}. The four formulations most directly affected by this paper are summarized below.

\begin{center}
\footnotesize
\renewcommand{\arraystretch}{1.18}
\begin{tabularx}{\textwidth}{@{}>{\raggedright\arraybackslash\bfseries}p{0.16\textwidth}>{\raggedright\arraybackslash}p{0.34\textwidth}YY@{}}
\toprule
Area & Established formulation & Status before this work & Consequence here \\
\midrule
Horn logic and AI
& Characteristic Models Identification (CMI), forward computation (CCM), and reverse Structure Identification (SID)
& Khardon proved the three tasks output-polynomially equivalent and left the exact complexity of CMI open in 1995~\cite{Khardon1995}.
& CMI is coNP-complete; CCM and SID are not output-polynomial unless $\mathrm P=\mathrm{NP}$. \\
\addlinespace
Closure systems and convex geometries
& Enumerate all irreducible closed sets from an implicational base
& The problem was described as ``widely open'' at ISAAC 2025, even for acyclic convex geometries~\cite{DefrainOhanaVilmin2025}.
& No output-polynomial algorithm exists unless $\mathrm P=\mathrm{NP}$, even for acyclic premise-size-three bases. \\
\addlinespace
Formal Concept Analysis
& Enumerate all pseudo-intents, equivalently compute the Duquenne--Guigues basis
& Unrestricted output-polynomial enumeration was explicitly left open in 2013~\cite{BabinKuznetsov2013}; recognition of one pseudo-intent was already coNP-complete~\cite{BabinKuznetsov2010}.
& All pseudo-intents and the DG basis are not output-polynomially computable unless $\mathrm P=\mathrm{NP}$. \\
\addlinespace
Relational databases
& FD--Relation Equivalence; generate the premises of a minimum FD cover from a relation
& Equivalence was highlighted as open in 1990~\cite{GottlobLibkin1990}; unrestricted premise enumeration was left open in 2013~\cite{BabinKuznetsov2013}.
& Equivalence is coNP-complete; minimum-cover premises are not output-polynomially generable unless $\mathrm P=\mathrm{NP}$. \\
\bottomrule
\end{tabularx}
\end{center}

These are not four unrelated reductions. They share a common closure-system core, and established equivalences make completeness of the canonical representation the algorithmic bottleneck.

\begin{center}
\setlength{\fboxsep}{8pt}
\fbox{\begin{minipage}{0.92\textwidth}
\textbf{Main result in plain language.}
Even after every entry in a proposed canonical list has been individually verified as correct, deciding whether the list contains \emph{all} canonical objects is coNP-complete.
\end{minipage}}
\end{center}

\subsection{The common mathematical core}

A finite set of implications has an associated closure operator. Starting from a set of facts, attributes, or database columns, one repeatedly applies rules until no further element is forced. The fixed points form a \emph{closure system}, a family of subsets closed under intersection.

There are two complementary representations of the same closure system. An \emph{implicational base} is operational: a rule $A\to b$ says that whenever all elements of $A$ are present, $b$ must also be present. The proper \emph{meet-irreducible closed sets} form a canonical extensional representation. A meet-irreducible is a maximal closed counterexample to some consequence: it omits an element $y$, but every strict closed extension contains $y$. Every closed set is recovered as an intersection of meet-irreducibles.

The terminology changes by field. In definite Horn logic, the rule side consists of Horn clauses and the canonical model side consists of characteristic models. In Formal Concept Analysis, closure systems are the intent families of binary contexts, while pseudo-intents index the canonical Duquenne--Guigues implication basis. In database theory, the rules are functional dependencies and canonical extensional representations are expressed through Armstrong relations and agreement sets. In acyclic convex geometries, acyclic implications define the closure operator and irreducible closed sets describe its boundary.

Formally, for an implicational base $\Sigma$, let $\C_\Sigma$ be its closed-set family. For a supplied family $\M$, let $\Int(\M)$ be the family of all intersections of members of $\M$. We study
\[
  \C_\Sigma=\Int(\M)?
\]
This equality question is the decision counterpart of translating between the rule representation and the complete irreducible representation.

\subsection{Main theorem and consequences}

Our main theorem gives a negative classification of this completeness problem.

\begin{informaltheorem}
Given an implicational base $\Sigma$ and a family $\M$ of subsets of the same ground set, deciding whether
\[
  \C_\Sigma=\Int(\M)
\]
is coNP-complete. Hardness holds even when $\Sigma$ is acyclic with premises of size at most three, every member of $\M$ is a genuine meet-irreducible of $\C_\Sigma$, and
\[
  \M=\MI(\Int(\M)).
\]
\end{informaltheorem}

The two promises on $\M$ explain the strength of the result. Without them, a no-instance could be rejected for a local reason: a listed set might not be closed, might not be irreducible, or might simply be redundant. Our construction excludes all such failures. Every listed set is correct for the rule-defined system, and the supplied family is a canonical irreducible representation of the closure system that it generates. What remains hard is the global question: is the list complete for $\C_\Sigma$?

The main theorem yields four headline consequences.
\begin{enumerate}[label=\textup{(C\arabic*)}]
  \item \textbf{Horn logic.} Characteristic Models Identification for definite Horn theories is coNP-complete even when every supplied model is genuinely characteristic. By Khardon's equivalence, neither direction of the Horn/characteristic-model translation admits an output-polynomial algorithm unless $\mathrm P=\mathrm{NP}$.
  \item \textbf{Enumeration and convex geometries.} Irreducible Closed Sets Enumeration has no output-polynomial algorithm unless $\mathrm P=\mathrm{NP}$, even for acyclic premise-size-three bases and hence for acyclic convex geometries.
  \item \textbf{Formal Concept Analysis.} All pseudo-intents of a formal context cannot be enumerated in output-polynomial time unless $\mathrm P=\mathrm{NP}$. Equivalently, the Duquenne--Guigues basis cannot be computed output-polynomially in general.
  \item \textbf{Functional dependencies.} The premises of a minimum FD cover cannot be generated in output-polynomial time from an explicit relation unless $\mathrm P=\mathrm{NP}$. Moreover, FD--Relation Equivalence is coNP-complete, even when the relation already satisfies the supplied FD set.
\end{enumerate}

The only new core proof is the closure-representation theorem. We use Khardon's established Horn equivalences without repeating their proofs. The appendices contain only the transfers that are not immediate from terminology: the output-size argument for pseudo-intents and the agreement-set construction for FD--Relation Equivalence.

\subsection{Proof idea and scope}

The reduction is from $3$-\textsc{UNSAT}. For every variable $z_i$ we create two choice elements $p_i,q_i$. For every clause $C_k$, the choices that jointly falsify $C_k$ trigger an initial gate $g_{k,0}$. A serial chain then tests, stage by stage, that at least one choice has been made for every variable; reaching its terminal gate derives a distinguished target $x$.

\begin{figure}[t]
\centering
\begin{tikzpicture}[
  x=1cm,y=1cm,
  every node/.style={font=\small},
  gate/.style={draw,rounded corners=1.5pt,minimum width=1.15cm,minimum height=0.58cm,inner sep=2pt},
  arr/.style={-{Latex[length=2.2mm]},thick},
  lab/.style={midway,above=3pt,fill=white,inner sep=1.2pt,font=\footnotesize}
]
  \node[gate] (a) at (0,0) {$A_k$};
  \node[gate] (g0) at (2.15,0) {$g_{k,0}$};
  \node[gate] (g1) at (4.30,0) {$g_{k,1}$};
  \node (dots) at (6.15,0) {$\cdots$};
  \node[gate] (gn) at (8.00,0) {$g_{k,n}$};
  \node[gate] (x) at (10.15,0) {$x$};

  \draw[arr] (a) -- node[lab,align=center] {clause\\falsified} (g0);
  \draw[arr] (g0) -- node[lab] {$p_1$ or $q_1$} (g1);
  \draw[arr] (g1) -- (dots);
  \draw[arr] (dots) -- node[lab] {$p_n$ or $q_n$} (gn);
  \draw[arr] (gn) -- node[lab] {trigger} (x);
\end{tikzpicture}
\caption{One clause chain. The base rule activates the chain exactly when the clause is falsified; a present choice at each variable stage propagates the chain to the target $x$.}
\label{fig:chain}
\end{figure}
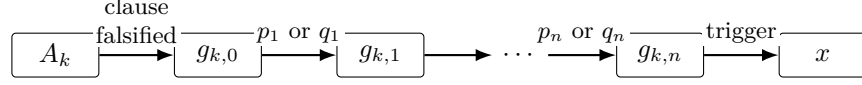

The supplied meet-irreducibles contain a polynomially classifiable family for every choice and every auxiliary gate, together with a fixed list of maximal $x$-avoiders. These objects are genuine meet-irreducibles for every formula. If the formula is unsatisfiable, the fixed list exhausts the maximal $x$-avoiders. If the formula is satisfiable, a satisfying assignment forms a closed $x$-free set that lies in none of the listed target candidates; extending it maximally produces a further meet-irreducible not present in the supplied list. The main technical work is a recurrence that classifies all auxiliary-gate cuts without circular reasoning.

The result does not settle hypergraph transversal enumeration, which remains an important special case with a different complexity status. It also does not contradict the bounded-degree algorithms of Defrain, Ohana, and Vilmin: our construction is acyclic but does not have bounded occurrence degree. Finally, the coNP-completeness of recognizing one pseudo-intent is prior work~\cite{BabinKuznetsov2010}; our new FCA consequence concerns unrestricted total enumeration of \emph{all} pseudo-intents.

The remainder of the paper proves the construction, states the cross-field consequences, and gives the two non-obvious transfer arguments in appendices.

\section{Preliminaries}

Let $U$ be a finite set. An implication over $U$ is an expression $A\to b$, where $A\subseteq U$ and $b\in U$. A finite implicational base $\Sigma$ defines the closure operator $\cl_\Sigma$ obtained by repeatedly adding $b$ whenever $A\to b\in\Sigma$ and $A$ is already present. Its closed-set family is
\[
  \C_\Sigma=\set{C\subseteq U\suchthat \cl_\Sigma(C)=C}.
\]
For a family $\mathcal F\subseteq 2^U$, put
\[
  \Int(\mathcal F)
  =\set{\bigcap\mathcal A\suchthat \mathcal A\subseteq\mathcal F},
\]
where the intersection of the empty family is $U$.

A proper closed set $M\in\C\setminus\{U\}$ is \emph{meet-irreducible} if
\[
  M=C_1\cap C_2,\qquad C_1,C_2\in\C,
\]
implies $M=C_1$ or $M=C_2$. We write $\MI(\C)$ for the family of proper meet-irreducibles.

\begin{problem}[Imp--MI Equality]
\label{prob:equality}
Given a finite set $U$, an implicational base $\Sigma$ over $U$, and a family $\M\subseteq 2^U$, decide whether
\[
  \C_\Sigma=\Int(\M).
\]
\end{problem}

\begin{lemma}[Membership in an intersection-generated family]
\label{lem:int-membership}
For $C\subseteq U$, define
\[
  J_\M(C)=\bigcap\set{M\in\M\suchthat C\subseteq M},
\]
again interpreting an empty intersection as $U$. Then
\[
  C\in\Int(\M)\quad\Longleftrightarrow\quad J_\M(C)=C.
\]
Consequently, membership in $\Int(\M)$ is decidable in polynomial time.
\end{lemma}

\begin{proof}
If $C=\bigcap\mathcal A$ for some $\mathcal A\subseteq\M$, then every member of $\mathcal A$ contains $C$, and hence
\[
  C\subseteq J_\M(C)\subseteq\bigcap\mathcal A=C.
\]
Conversely, $J_\M(C)$ is itself an intersection of members of $\M$.
\end{proof}

\begin{lemma}[Maximal omitters]
\label{lem:max-omitter}
Let $\C$ be a finite closure system on $U$, and let $M\in\C\setminus\{U\}$. Then $M\in\MI(\C)$ if and only if there is some $y\in U\setminus M$ such that $M$ is inclusion-maximal among the closed sets omitting $y$.
\end{lemma}

\begin{proof}
Suppose $M$ is maximal among the closed sets omitting $y$. If $M=C_1\cap C_2$ with $C_1,C_2\in\C$ both strictly larger than $M$, maximality gives $y\in C_1\cap C_2=M$, a contradiction. Thus $M$ is meet-irreducible.

Conversely, let $M$ be meet-irreducible. Since $\C$ is finite and $M\neq U$, $M$ has a cover. It has only one. Indeed, if $D_1,D_2$ were distinct covers of $M$, then $D_1\cap D_2$ would be a closed set between $M$ and each $D_i$. The cover property would force $D_1\cap D_2=M$, contradicting meet-irreducibility. Let $M^+$ be the unique cover and choose $y\in M^+\setminus M$. Every strict closed superset of $M$ contains a cover of $M$, hence contains $M^+$ and therefore $y$. Thus $M$ is maximal among the closed sets omitting $y$.
\end{proof}

\begin{lemma}[One-element maximality test]
\label{lem:one-element}
Let $C$ be $\Sigma$-closed and let $y\notin C$. Then $C$ is maximal among the closed sets omitting $y$ if and only if
\[
  y\in\cl_\Sigma(C\cup\{z\})
  \qquad\text{for every }z\in U\setminus C.
\]
\end{lemma}

\begin{proof}
The forward implication is immediate because $\cl_\Sigma(C\cup\{z\})$ is a strict closed extension of $C$. Conversely, if $D$ is any strict closed extension of $C$, choose $z\in D\setminus C$. Then
\[
  y\in\cl_\Sigma(C\cup\{z\})\subseteq D.
\]
\end{proof}

\begin{lemma}[Generation by meet-irreducibles]
\label{lem:mi-generation}
Every finite closure system satisfies
\[
  \C=\Int(\MI(\C)).
\]
\end{lemma}

\begin{proof}
Fix $C\in\C$. For every $y\in U\setminus C$, extend $C$ to a maximal closed set $M_y$ still omitting $y$. By \cref{lem:max-omitter}, $M_y$ is meet-irreducible. Since $C\subseteq M_y$ and $y\notin M_y$,
\[
  C=\bigcap_{y\in U\setminus C}M_y.
\]
For $C=U$, this is the empty intersection. The reverse inclusion follows because $\C$ is intersection-closed.
\end{proof}

\begin{lemma}[Subfamilies of ambient meet-irreducibles]
\label{lem:subfamily}
Let $\C$ be a finite closure system and let $\mathcal A\subseteq\MI(\C)$. Then
\[
  \MI(\Int(\mathcal A))=\mathcal A.
\]
\end{lemma}

\begin{proof}
Because $\C$ is intersection-closed, $\Int(\mathcal A)\subseteq\C$. If $A\in\mathcal A$ and $A=D_1\cap D_2$ with $D_1,D_2\in\Int(\mathcal A)$, meet-irreducibility of $A$ in $\C$ gives $A=D_1$ or $A=D_2$. Thus every member of $\mathcal A$ is meet-irreducible in $\Int(\mathcal A)$.

Conversely, let $N$ be a proper meet-irreducible of $\Int(\mathcal A)$. Choose an inclusion-minimal nonempty subfamily $\mathcal B\subseteq\mathcal A$ with $N=\bigcap\mathcal B$. If $|\mathcal B|\ge2$, choose $B\in\mathcal B$. Minimality gives
\[
  N\subsetneq B,
  \qquad
  N\subsetneq\bigcap(\mathcal B\setminus\{B\}),
\]
and these two strict supersets intersect to $N$, a contradiction. Hence $|\mathcal B|=1$ and $N\in\mathcal A$.
\end{proof}

The dependency digraph $G_\Sigma$ has an arc $a\to b$ whenever $a$ belongs to the premise of an implication with conclusion $b$. Let $\Anc_\Sigma(y)$ be the set consisting of $y$ and all vertices from which a directed path reaches $y$.

\begin{lemma}[Dependency locality]
\label{lem:locality}
For every $S\subseteq U$ and $y\in U$,
\[
  y\in\cl_\Sigma(S)
  \quad\Longleftrightarrow\quad
  y\in\cl_\Sigma\bigl(S\cap\Anc_\Sigma(y)\bigr).
\]
In particular, if $z\notin\Anc_\Sigma(y)$, inserting or deleting $z$ from the initial set does not affect derivability of $y$.
\end{lemma}

\begin{proof}
Only the forward implication needs proof. Take a forward-chaining derivation of $y$ from $S$ and delete every derived element outside $\Anc_\Sigma(y)$. If $b\in\Anc_\Sigma(y)$ was obtained from $A\to b$, then every $a\in A$ has an arc to $b$ and therefore also lies in $\Anc_\Sigma(y)$. Hence the filtered sequence is still a derivation of $y$, now from $S\cap\Anc_\Sigma(y)$.
\end{proof}

\begin{corollary}
\label{cor:nonancestor}
If $C$ is maximal closed among the sets omitting $y$, then every $z\notin\Anc_\Sigma(y)$ belongs to $C$.
\end{corollary}

\begin{proof}
Otherwise \cref{lem:locality} would imply that $\cl_\Sigma(C\cup\{z\})$ still omits $y$, contradicting \cref{lem:one-element}.
\end{proof}

\begin{proposition}
\label{prop:conp-membership}
\problemname{Imp--MI Equality} belongs to coNP.
\end{proposition}

\begin{proof}
A no-certificate is a set $C\subseteq U$ belonging to exactly one of $\C_\Sigma$ and $\Int(\M)$. Membership in $\C_\Sigma$ is checked by scanning the implications, and membership in $\Int(\M)$ is checked by \cref{lem:int-membership}.
\end{proof}

\section{The reduction}

We reduce from $3$-\textsc{UNSAT}, which is coNP-complete~\cite{GareyJohnson1979}. We use clauses of size at most three. We may assume, without loss of generality, that the formula has at least one variable and at least one clause, and that every clause is nonempty, contains no repeated literal, and is not tautological. Here is an explicit polynomial normalization. If the input contains an empty clause, map it to the fixed unsatisfiable formula $(y)\wedge(\neg y)$. Otherwise delete tautological clauses, delete repeated copies of a literal inside each remaining clause, and conjoin a fresh unit clause $(y)$. Finally, relabel the variables that actually occur. The resulting formula is satisfiable exactly when the input is satisfiable and satisfies all the stated conditions.

Let
\[
  \varphi=C_1\wedge\cdots\wedge C_m
\]
be a normalized formula on variables $z_1,\ldots,z_n$, where $m,n\ge1$. For each variable introduce two choice elements
\[
  p_i\quad\text{and}\quad q_i,
\]
representing $z_i=1$ and $z_i=0$, respectively. Define the choice that falsifies a literal by
\[
  f(z_i)=q_i,
  \qquad
  f(\neg z_i)=p_i.
\]
For each clause $C_k$, put
\[
  A_k=\set{f(\ell)\suchthat \ell\in C_k}.
\]
Then $1\le |A_k|\le3$, and $A_k$ contains at most one of $p_i,q_i$ for every $i$. A truth assignment $\alpha$ is represented by
\[
  S_\alpha
  =\set{p_i\suchthat \alpha(z_i)=1}
   \cup
   \set{q_i\suchthat \alpha(z_i)=0}.
\]
By construction,
\begin{equation}
  \alpha\text{ falsifies }C_k
  \quad\Longleftrightarrow\quad
  A_k\subseteq S_\alpha.
  \label{eq:falsifies}
\end{equation}

For every clause $k$, introduce a chain
\[
  g_{k,0},g_{k,1},\ldots,g_{k,n},
\]
and introduce one final target $x$. The ground set is
\[
  U=
  \set{p_i,q_i:1\le i\le n}
  \cup
  \set{g_{k,j}:1\le k\le m,\ 0\le j\le n}
  \cup\{x\}.
\]
The base $\Sigma_\varphi$ consists of
\begin{align}
  A_k &\to g_{k,0} &&(1\le k\le m), \label{eq:base}\\
  \{g_{k,j-1},p_j\} &\to g_{k,j},
  &\{g_{k,j-1},q_j\}&\to g_{k,j}
  &&(1\le k\le m,\ 1\le j\le n), \label{eq:step}\\
  \{g_{k,n}\} &\to x &&(1\le k\le m). \label{eq:terminal}
\end{align}
Every conclusion is a singleton and every premise has size at most three. The rank function
\[
  \rho(p_i)=\rho(q_i)=0,
  \qquad
  \rho(g_{k,j})=j+1,
  \qquad
  \rho(x)=n+2
\]
strictly increases along every dependency arc, so the dependency digraph is acyclic. For the rest of the proof, write $\cl=\cl_{\Sigma_\varphi}$ and $\C_\varphi=\C_{\Sigma_\varphi}$.

\subsection{The supplied family}

We describe the supplied closed sets by their complements, called cuts.

For each $i$, include the choice complements
\[
  P_i=U\setminus\{p_i\},
  \qquad
  Q_i=U\setminus\{q_i\}.
\]
For a clause $k$, a stage $0\le j\le n$, and $a\in A_k$, define
\[
  E_{k,j}^{a}=\{a,g_{k,0},g_{k,1},\ldots,g_{k,j}\}.
\]
For $1\le t\le j$, define
\[
  E_{k,j}^{t}=\{p_t,q_t,g_{k,t},g_{k,t+1},\ldots,g_{k,j}\}.
\]
We include the complements of all these gate cuts.

Finally, for $1\le i\le n$, define
\[
  E_i^x
  =\{x,p_i,q_i\}
   \cup
   \set{g_{k,j}:1\le k\le m,\ i\le j\le n},
  \qquad
  M_i^x=U\setminus E_i^x.
\]
The supplied family is
\begin{align*}
  \M_\varphi={}&
  \set{P_i,Q_i:1\le i\le n}\\
  &\cup\set{U\setminus E_{k,j}^a:
      1\le k\le m,\ 0\le j\le n,\ a\in A_k}\\
  &\cup\set{U\setminus E_{k,j}^t:
      1\le k\le m,\ 1\le t\le j\le n}\\
  &\cup\set{M_i^x:1\le i\le n}.
\end{align*}

\section{Classification of the supplied meet-irreducibles}

\subsection{Choice elements}

\begin{lemma}
\label{lem:choice}
For every $i$, $P_i$ is the unique maximal closed set omitting $p_i$, and $Q_i$ is the unique maximal closed set omitting $q_i$.
\end{lemma}

\begin{proof}
No implication concludes a choice element. Hence $P_i$ and $Q_i$ are closed, and every set omitting the relevant choice is contained in the corresponding complement.
\end{proof}

\subsection{Auxiliary gates}

Fix a clause index $k$. For $0\le j\le n$, let $\B_{k,j}$ be the family of complements of maximal closed sets omitting $g_{k,j}$.

\begin{lemma}[Base gate]
\label{lem:base-gate}
For every $k$,
\[
  \B_{k,0}=\set{\{g_{k,0},a\}:a\in A_k}.
\]
\end{lemma}

\begin{proof}
Let $C$ be maximal closed subject to $g_{k,0}\notin C$. The rule $A_k\to g_{k,0}$ forces at least one $a\in A_k$ to be absent. At most one member of $A_k$ is absent: if distinct $a,b\in A_k$ were both missing, then $\cl(C\cup\{a\})$ would still miss $b$, so the only rule concluding $g_{k,0}$ would remain blocked, contradicting maximality. Fix the unique missing $a$.

No element other than $g_{k,0}$ and $a$ can be missing. Adding any further missing element would leave $a$ absent and would still not derive $g_{k,0}$. Hence $U\setminus C=\{g_{k,0},a\}$.

Conversely, $U\setminus\{g_{k,0},a\}$ is closed because the base rule is blocked by $a$. Restoring $g_{k,0}$ includes it directly, and restoring $a$ completes $A_k$ and derives it. Thus the complement is maximal $g_{k,0}$-free.
\end{proof}

\begin{lemma}[Serial recurrence]
\label{lem:recurrence}
For every $1\le j\le n$,
\[
  \B_{k,j}
  =
  \set{\{g_{k,j},p_j,q_j\}}
  \cup
  \set{\{g_{k,j}\}\cup B:B\in\B_{k,j-1}}.
\]
\end{lemma}

\begin{proof}
Let $C$ be maximal closed subject to $g_{k,j}\notin C$. By \cref{cor:nonancestor}, every later gate on chain $k$, every gate on another chain, and $x$ belong to $C$.

First suppose $g_{k,j-1}\in C$. Closedness forces both $p_j$ and $q_j$ to be absent. No other element can be absent: after adding any further missing element, both choices would remain absent because no implication concludes a choice, so $g_{k,j}$ would remain underivable. Therefore
\[
  U\setminus C=\{g_{k,j},p_j,q_j\}.
\]
Conversely, the complement of this cut is closed and becomes $g_{k,j}$-containing after any omitted element is restored.

Now suppose $g_{k,j-1}\notin C$. At least one of $p_j,q_j$ belongs to $C$. Otherwise choose
\[
  r\in\{p_j,q_j\}\setminus A_k,
\]
which is possible because $A_k$ contains at most one member of the pair. The element $r$ is not an ancestor of $g_{k,j-1}$: it does not enter chain $k$ through the base rule, and its transition edge enters only the later gate $g_{k,j}$. By \cref{lem:locality}, adding $r$ still does not derive $g_{k,j-1}$, and hence does not derive $g_{k,j}$ either, contradicting maximality.

Put
\[
  D=C\cup\{g_{k,j}\}.
\]
The set $D$ is closed. Inserting $g_{k,j}$ can enable only later gates on the same chain, or $x$ when $j=n$, and all these elements already belong to $C$. The set $D$ still omits $g_{k,j-1}$. For every $z\notin D$, maximality of $C$ gives
\[
  g_{k,j}\in\cl(C\cup\{z\}).
\]
Since $z\neq g_{k,j}$, the first derivation of $g_{k,j}$ uses one of the two stage-$j$ rules and therefore derives $g_{k,j-1}$. Thus
\[
  g_{k,j-1}\in\cl(D\cup\{z\}).
\]
By \cref{lem:one-element}, $D$ is maximal $g_{k,j-1}$-free. Hence $U\setminus D\in\B_{k,j-1}$ and
\[
  U\setminus C=\{g_{k,j}\}\cup(U\setminus D).
\]
This proves the forward inclusion.

For the converse, let $B\in\B_{k,j-1}$ and put $D=U\setminus B$. Then $D$ is maximal closed subject to omitting $g_{k,j-1}$. Since $g_{k,j}$ is not an ancestor of $g_{k,j-1}$, \cref{cor:nonancestor} gives $g_{k,j}\in D$. At least one of $p_j,q_j$ lies in $D$: if both were absent, the same choice of $r\in\{p_j,q_j\}\setminus A_k$ and \cref{lem:locality} would contradict maximality of $D$.

Set
\[
  C=D\setminus\{g_{k,j}\}
   =U\setminus\bigl(B\cup\{g_{k,j}\}\bigr).
\]
The set $C$ is closed. Indeed, if a rule had its premise in $C$ and its conclusion outside $C$, then closedness of $D$ would force that conclusion to be $g_{k,j}$. But every rule concluding $g_{k,j}$ requires the absent predecessor $g_{k,j-1}$.

It remains to prove maximality. Restoring $g_{k,j}$ includes it directly. If $z\in B$, maximality of $D$ gives
\[
  g_{k,j-1}\in\cl(D\cup\{z\}).
\]
The later gate $g_{k,j}$ is not an ancestor of $g_{k,j-1}$, so locality allows it to be deleted from the initial set:
\[
  g_{k,j-1}\in\cl(C\cup\{z\}).
\]
At least one of $p_j,q_j$ is already in $C$, and hence $g_{k,j}$ is derived. By \cref{lem:one-element}, $C$ is maximal $g_{k,j}$-free. This proves the converse inclusion.
\end{proof}

\begin{proposition}[Closed form of the serial cuts]
\label{prop:gate-cuts}
For every $0\le j\le n$,
\[
  \B_{k,j}
  =\set{E_{k,j}^a:a\in A_k}
   \cup
   \set{E_{k,j}^t:1\le t\le j}.
\]
\end{proposition}

\begin{proof}
Induct on $j$. The case $j=0$ is \cref{lem:base-gate}. Under the recurrence of \cref{lem:recurrence}, extending a base cut $\{g_{k,0},a\}$ through stages $1,\ldots,j$ yields
\[
  \{a,g_{k,0},\ldots,g_{k,j}\}=E_{k,j}^a.
\]
Starting with the new stage cut $\{p_t,q_t,g_{k,t}\}$ at stage $t$ and extending it through $t+1,\ldots,j$ yields
\[
  \{p_t,q_t,g_{k,t},\ldots,g_{k,j}\}=E_{k,j}^t.
\]
These are all branches of the recurrence.
\end{proof}

\subsection{The target sets}

\begin{lemma}
\label{lem:target-sets}
For every $1\le i\le n$, the set $M_i^x$ is closed and maximal among the closed sets omitting $x$.
\end{lemma}

\begin{proof}
The set $M_i^x$ contains every choice except $p_i,q_i$, contains every gate of stage below $i$, and contains no gate of stage at least $i$ and no $x$. The two missing choices block every chain at stage $i$, so no terminal gate and no $x$ is derived. Hence $M_i^x$ is closed.

Use \cref{lem:one-element}. Restoring $x$ includes it directly. Restoring $p_i$ or $q_i$ advances every already present $g_{k,i-1}$ to $g_{k,i}$; all later choices are present, so every chain reaches a terminal gate and derives $x$. The assumption $m\ge1$ guarantees that a chain exists. Finally, restoring a missing gate $g_{k,j}$ with $j\ge i$ starts chain $k$ at that stage, after which the later choices advance it to $g_{k,n}$ and then to $x$. Thus every proper closed extension contains $x$.
\end{proof}

By \cref{lem:max-omitter,lem:choice,prop:gate-cuts,lem:target-sets}, every member of $\M_\varphi$ is a genuine proper meet-irreducible of $\C_\varphi$, independently of the satisfiability of $\varphi$.

\section{The formula-dependent target class}

\begin{lemma}[Target classification]
\label{lem:target-classification}
The following are equivalent:
\begin{enumerate}[label=\textup{(\roman*)}]
  \item $\varphi$ is unsatisfiable;
  \item every closed set omitting $x$ is contained in some $M_i^x$;
  \item the maximal closed sets omitting $x$ are exactly $M_1^x,\ldots,M_n^x$.
\end{enumerate}
If $\varphi$ is satisfiable and $\alpha$ is a satisfying assignment, then $S_\alpha$ is closed, omits $x$, and is contained in no $M_i^x$.
\end{lemma}

\begin{proof}
The equivalence of (ii) and (iii) follows from \cref{lem:target-sets} and finiteness: every closed $x$-free set extends to a maximal one.

Assume first that $\varphi$ is unsatisfiable. Let $C$ be closed, let $x\notin C$, and suppose for contradiction that $C\nsubseteq M_i^x$ for every $i$.

Suppose $C$ contains a gate. Choose a gate $g_{k,r}\in C$ whose stage $r$ is maximum among all gates in $C$. Since $C$ omits $x$, we have $r<n$. For every $i>r$, noncontainment in $M_i^x$ gives an element of
\[
  C\cap(E_i^x\setminus\{x\})
  \subseteq
  C\cap\left(\{p_i,q_i\}\cup
  \set{g_{\ell,s}:s\ge i}\right).
\]
By maximality of $r$, no gate of stage at least $i$ belongs to $C$, so $p_i\in C$ or $q_i\in C$. Starting from $g_{k,r}$, these choices advance chain $k$ through stages $r+1,\ldots,n$ and derive $x$, a contradiction. Therefore $C$ contains no gate.

Now noncontainment in $M_i^x$ implies that $p_i\in C$ or $q_i\in C$ for every $i$. Choose exactly one present element from each pair, obtaining a complete assignment set $S_\alpha\subseteq C$. Since $\varphi$ is unsatisfiable, $\alpha$ falsifies some clause $C_k$. By \eqref{eq:falsifies}, $A_k\subseteq S_\alpha\subseteq C$, so the base rule derives $g_{k,0}$. A choice is present at every stage, so the chain reaches $g_{k,n}$ and derives $x$, again a contradiction. Thus (ii) holds.

Conversely, suppose $\varphi$ is satisfiable and let $\alpha$ satisfy it. No $A_k$ is contained in $S_\alpha$, so no initial gate is derived. Since $S_\alpha$ contains no gate, no transition or terminal rule fires. Hence $S_\alpha$ is closed and omits $x$. For each $i$, the set $S_\alpha$ contains exactly one of $p_i,q_i$, while $M_i^x$ contains neither, so $S_\alpha\nsubseteq M_i^x$. This disproves (ii).
\end{proof}

\begin{proposition}[Complete meet-irreducible classification]
\label{prop:classification}
Every member of $\M_\varphi$ belongs to $\MI(\C_\varphi)$. Moreover,
\[
  \M_\varphi=\MI(\C_\varphi)
  \quad\Longleftrightarrow\quad
  \varphi\text{ is unsatisfiable}.
\]
If $\varphi$ is satisfiable, at least one unlisted meet-irreducible is maximal among the closed sets omitting $x$.
\end{proposition}

\begin{proof}
Every ground-set element is a choice, a gate, or $x$. By \cref{lem:max-omitter}, every meet-irreducible is maximal among the closed sets omitting some ground-set element. The maximal omitters of choices are classified by \cref{lem:choice}; those of gates by \cref{prop:gate-cuts}; and those of $x$ by \cref{lem:target-classification}. This proves the equivalence.

If $\varphi$ is satisfiable, extend the closed assignment set $S_\alpha$ from \cref{lem:target-classification} to a maximal closed $x$-free set $N$. Then $N$ is meet-irreducible by \cref{lem:max-omitter}. It is not one of the $M_i^x$, because it contains $S_\alpha$ and hence one of $p_i,q_i$ for every $i$. It is not a non-target supplied set either, because every non-target member of $\M_\varphi$ contains $x$. Thus $N\notin\M_\varphi$.
\end{proof}

A useful promise now follows without a separate irredundancy argument.

\begin{corollary}[Self-exactness of the supplied family]
\label{cor:self-exact}
For every formula $\varphi$,
\[
  \M_\varphi=\MI(\Int(\M_\varphi)).
\]
\end{corollary}

\begin{proof}
By \cref{prop:classification}, $\M_\varphi\subseteq\MI(\C_\varphi)$. Apply \cref{lem:subfamily}.
\end{proof}

\section{Correctness and complexity}

\begin{theorem}[Correctness of the reduction]
\label{thm:correctness}
For every normalized $3$-CNF formula $\varphi$, the following are equivalent:
\begin{enumerate}[label=\textup{(\roman*)}]
  \item $\varphi$ is unsatisfiable;
  \item $\M_\varphi=\MI(\C_\varphi)$;
  \item $\C_\varphi=\Int(\M_\varphi)$.
\end{enumerate}
\end{theorem}

\begin{proof}
The equivalence of (i) and (ii) is \cref{prop:classification}. If (ii) holds, generation by meet-irreducibles gives
\[
  \C_\varphi
  =\Int(\MI(\C_\varphi))
  =\Int(\M_\varphi),
\]
so (iii) holds.

It remains to show directly that satisfiability rules out (iii). Let $\alpha$ satisfy $\varphi$. By \cref{lem:target-classification}, $S_\alpha$ is a closed set omitting $x$ and is contained in no target set $M_i^x$. Every non-target member of $\M_\varphi$ contains $x$. If $S_\alpha$ were an intersection of supplied sets, every factor would contain $S_\alpha$, so no target factor could occur. Every factor would therefore contain $x$, and so would their intersection. The empty intersection is $U$ and also contains $x$. This contradicts $x\notin S_\alpha$. Hence $S_\alpha\notin\Int(\M_\varphi)$ and (iii) fails.
\end{proof}

\begin{lemma}[Polynomial size]
\label{lem:size}
The instance $(U,\Sigma_\varphi,\M_\varphi)$ is computable in polynomial time and has polynomial encoding length.
\end{lemma}

\begin{proof}
The ground set and base have sizes
\[
  |U|=2n+m(n+1)+1,
  \qquad
  |\Sigma_\varphi|=m+2mn+m=2m(n+1).
\]
The number of supplied sets is
\begin{align*}
  |\M_\varphi|
  &=3n+\sum_{k=1}^{m}\sum_{j=0}^{n}(|A_k|+j)\\
  &=3n+\sum_{k=1}^{m}
    \left((n+1)|A_k|+\frac{n(n+1)}2\right)\\
  &=O(mn^2).
\end{align*}
Each supplied set is encoded by a bit vector of length $|U|=O(mn)$, so the total output encoding is polynomial in the input size.
\end{proof}

\begin{theorem}[Main theorem]
\label{thm:main}
\problemname{Imp--MI Equality} is coNP-complete. CoNP-hardness holds even on instances satisfying all of the following simultaneously:
\begin{enumerate}[label=\textup{(\alph*)}]
  \item every implication has a singleton conclusion and a premise of size at most three;
  \item the dependency digraph is acyclic;
  \item every supplied set belongs to $\MI(\C_\Sigma)$;
  \item the supplied family satisfies
  \[
    \M=\MI(\Int(\M)).
  \]
\end{enumerate}
Equivalently, every supplied set passes the natural validity tests; coNP-hardness remains in the global question of whether the canonical list is complete.
\end{theorem}

\begin{proof}
Membership in coNP is \cref{prop:conp-membership}. The construction is polynomial by \cref{lem:size}, and \cref{thm:correctness} gives
\[
  \varphi\in 3\text{-}\problemname{UNSAT}
  \quad\Longleftrightarrow\quad
  \C_{\Sigma_\varphi}=\Int(\M_\varphi).
\]
The structural promises follow from the construction, \cref{prop:classification}, and \cref{cor:self-exact}.
\end{proof}

\section{Consequences across representations}
\label{sec:consequences}

\subsection{Identification and characteristic models}

The main theorem immediately yields the corresponding exact identification problem.

\begin{corollary}[Meet-irreducible identification]
\label{cor:mi-identification-final}
Given an implicational base $\Sigma$ and a family $\M$, deciding whether
\[
  \M=\MI(\C_\Sigma)
\]
is coNP-complete. Hardness holds under the promises of \cref{thm:main}, including $\M\subseteq\MI(\C_\Sigma)$.
\end{corollary}

For completeness, membership in coNP has a short certificate: either a listed set is not a proper meet-irreducible, witnessed by nonclosedness or by two strict closed supersets whose intersection is that set, or an unlisted meet-irreducible is supplied together with an attached element and the one-element maximality test of \cref{lem:one-element}.

Under the subset representation of assignments, the characteristic models of a definite Horn theory are precisely the all-ones model together with its proper meet-irreducible models. Khardon's CMI problem takes a Horn theory $H$ and a family $\G$ of satisfying assignments and asks whether every characteristic model of $H$ belongs to $\G$~\cite{Khardon1995}. Adding the all-ones model to our supplied family gives the next result.

\begin{corollary}[Characteristic Models Identification]
\label{cor:cmi-final}
Characteristic Models Identification for definite Horn theories is coNP-complete. Hardness holds for acyclic theories with bodies of size at most three, even under the promise that every supplied model is genuinely characteristic; the hard question is whether the supplied list is complete.
\end{corollary}

Khardon proved that Characteristic Models Identification, computing all characteristic models, and reconstructing a Horn representation from characteristic models are equivalent under output-polynomial reductions~\cite{Khardon1995}. Therefore \cref{cor:cmi-final} supplies the missing lower bound in that equivalence.

\begin{corollary}[Horn translation barrier]
\label{cor:horn-output}
Unless $\mathrm P=\mathrm{NP}$, there is no output-polynomial algorithm for either of the following general translation tasks:
\begin{enumerate}[label=\textup{(\roman*)}]
  \item computing all characteristic models of a Horn theory;
  \item constructing a Horn representation from its characteristic models.
\end{enumerate}
\end{corollary}

\subsection{Irreducible closed-set enumeration}

Defrain, Ohana, and Vilmin call the task of enumerating all irreducible closed sets from an implicational base \problemname{ICSEnum}~\cite{DefrainOhanaVilmin2025}. The output may be exponentially larger than the base, so output-polynomial total time is the baseline tractability notion.

\begin{corollary}[Irreducible Closed Sets Enumeration]
\label{cor:icsenum}
Unless $\mathrm P=\mathrm{NP}$, \problemname{ICSEnum} has no output-polynomial algorithm. The lower bound already holds for acyclic implicational bases with singleton conclusions and premises of size at most three; consequently it holds for acyclic convex geometries.
\end{corollary}

The output-sensitive argument is given in \cref{app:enumeration}. The convex-geometric statement uses the standard observation that an acyclic base without empty premises defines a convex geometry: dependency locality turns a derivation of $u$ from a closed set plus $v$ into a directed path from $v$ to $u$, so acyclicity forbids the reverse derivation. In particular, \cref{cor:icsenum} rules out polynomial delay and incremental-polynomial time in the unrestricted acyclic class. It leaves intact the bounded premise- and conclusion-degree algorithms of~\cite{DefrainOhanaVilmin2025}.

\subsection{Pseudo-intents and the Duquenne--Guigues basis}

For a formal context $\K$, the Duquenne--Guigues basis contains one implication
\[
  P\to P''
\]
for every pseudo-intent $P$. It has minimum cardinality among all implication bases of the same closure system~\cite{GuiguesDuquenne1986,Kuznetsov2004}. Thus enumerating all pseudo-intents and constructing the canonical basis are polynomially interconvertible.

Recognition of a single pseudo-intent is coNP-complete~\cite{BabinKuznetsov2010}. Earlier enumeration work ruled out polynomial delay in prescribed orders and output-polynomial enumeration for restricted extremal subfamilies~\cite{DistelSertkaya2011,BabinKuznetsov2013}, but left unrestricted output-polynomial enumeration of all pseudo-intents open.

\begin{corollary}[Pseudo-intents and the DG basis]
\label{cor:pseudo-output}
Unless $\mathrm P=\mathrm{NP}$, there is no output-polynomial algorithm that, given an arbitrary finite formal context,
\begin{enumerate}[label=\textup{(\roman*)}]
  \item enumerates all pseudo-intents; or equivalently,
  \item computes its Duquenne--Guigues basis.
\end{enumerate}
\end{corollary}

The derivation uses Khardon's output-sensitive equivalence and the cardinality minimality of the DG basis. The only nonautomatic point is that the canonical basis is polynomially bounded by a shortest permissible Horn output; this is recorded in \cref{app:pseudo}.

\subsection{Functional dependencies and Armstrong relations}

A functional dependency $X\to Y$ holding in a relation is the database analogue of an implication holding in a context. Babin and Kuznetsov showed that pseudo-intents correspond to premises of minimum covers of the dependencies valid in a relation, and explicitly identified unrestricted output-polynomial generation as an open problem~\cite{BabinKuznetsov2013}.

\begin{corollary}[Minimum covers from relations]
\label{cor:fd-cover}
Unless $\mathrm P=\mathrm{NP}$, the premises of a minimum functional-dependency cover of an explicit relation cannot be generated in output-polynomial time. In particular, no output-polynomial algorithm can compute a minimum cover from an arbitrary relation.
\end{corollary}

This should be contrasted with Maier's polynomial-time algorithm for minimizing an FD set that is already supplied explicitly~\cite{Maier1979}. The intractability concerns discovering the implicit dependency theory of a relation, not simplifying a given rule set.

We also obtain a decision result that does not follow by terminology alone. Let $F_R$ denote all FDs valid in a relation $R$, and let $F^+$ denote the consequences of an FD set $F$.

\begin{problem}[FD--Relation Equivalence]
Given a finite relation $R$ and a finite FD set $F$ on the same attributes, decide whether
\[
  F_R=F^+.
\]
Equivalently, decide whether $R$ is an Armstrong relation for $F$.
\end{problem}

Gottlob and Libkin highlighted the complexity of this problem as open~\cite{GottlobLibkin1990}. An explicit agreement-set realization of our hard instances gives the following theorem; the reduction is in \cref{app:fd}.

\begin{theorem}[FD--Relation Equivalence]
\label{thm:fd-equivalence}
\problemname{FD--Relation Equivalence} is coNP-complete. Hardness holds even under the promise
\[
  F^+\subseteq F_R,
\]
that is, even when $R$ already satisfies every dependency in $F$ and equivalence can fail only because $R$ satisfies additional dependencies.
\end{theorem}

\section{Conclusion}

The proof isolates a simple source of hardness: completeness of a canonical representation. The supplied sets are all valid, and collectively form an exact irreducible representation of another closure system. Nevertheless, deciding whether further canonical objects exist is coNP-complete.

This completeness barrier propagates through several classical equivalences. It closes the unrestricted output-polynomial question for irreducible closed sets, characteristic models, pseudo-intents, the Duquenne--Guigues basis, and minimum FD-cover premises; it also classifies FD--Relation Equivalence. At the same time, the result preserves a meaningful tractability frontier: bounded occurrence degree can still support efficient enumeration, while acyclicity and bounded premise size alone do not.

\appendix

\section{Output-sensitive consequences}
\label{app:enumeration}

We include the output-sensitive details because a decision lower bound does not automatically imply an enumeration lower bound when no-instances may have exponentially larger outputs.

\begin{lemma}[Completeness testing from an output-polynomial enumerator]
\label{lem:enumeration-transfer}
Let $L$ be a coNP-hard language. Suppose a polynomial-time reduction maps every input $w$ to an enumeration instance $I_w$ and an explicitly listed family $S_w$ such that:
\begin{enumerate}[label=\textup{(\roman*)}]
  \item $S_w$ is contained in the solution set $\operatorname{Sol}(I_w)$;
  \item $|I_w|+|S_w|$ is polynomial in $|w|$;
  \item $w\in L$ if and only if $\operatorname{Sol}(I_w)=S_w$.
\end{enumerate}
Then the enumeration problem has no output-polynomial algorithm unless $\mathrm P=\mathrm{NP}$.
\end{lemma}

\begin{proof}
Assume an enumerator runs in time $p(|I|+|\operatorname{Sol}(I)|)$ for a fixed polynomial $p$. On a yes-instance $w\in L$, the complete output equals the polynomial-size family $S_w$, so the enumerator terminates within a polynomial bound in $|w|$. Simulate it for that bound. If it has not terminated, reject. If it terminates, compare its complete output with $S_w$ and accept exactly when they agree. On a no-instance, either the run has not completed within the bound or its completed output differs from $S_w$. This decides the coNP-hard language $L$ in polynomial time, implying $\mathrm P=\mathrm{coNP}$ and hence $\mathrm P=\mathrm{NP}$.
\end{proof}

\begin{proof}[Proof of \cref{cor:icsenum}]
Apply \cref{lem:enumeration-transfer} to the reduction of \cref{thm:main}, with $I_\varphi=\Sigma_\varphi$ and $S_\varphi=\M_\varphi$. By \cref{prop:classification}, every supplied set is a genuine irreducible closed set, and equality with the full output holds exactly when $\varphi$ is unsatisfiable. The reduction has polynomial size and its bases are acyclic with premises of size at most three.
\end{proof}

\section{Pseudo-intents and the output measure}
\label{app:pseudo}

Let $\G\subseteq 2^U$ contain $U$, and form the context
\[
  \K_\G=(\G,U,\in),
\]
whose objects are the sets in $\G$ and whose object intent corresponding to $G\in\G$ is $G$ itself. Its intent closure system is exactly $\Int(\G)$.

Suppose that an output-polynomial algorithm enumerates all pseudo-intents of an arbitrary context. From each pseudo-intent $P$ we compute $P''$ in polynomial time and obtain the Duquenne--Guigues basis of $\K_\G$. The DG basis has one implication per pseudo-intent and minimum cardinality among all bases~\cite{GuiguesDuquenne1986}. If a shortest Horn representation of $\Int(\G)$ has encoded length $s$, it contains at most $s$ clauses, so the DG basis has at most $s$ implications. Each premise and conclusion is a subset of $U$ and has an $O(|U|)$-bit encoding. Thus the complete pseudo-intent list and the DG basis have length $O(|U|s)$; splitting set-valued conclusions into singleton conclusions costs at most another factor $|U|$.

Consequently, the hypothetical enumerator would solve Khardon's Structure Identification problem in time polynomial in the input and in a shortest permissible Horn output, at least on the definite-Horn instances used in our reduction. Khardon proved that CMI reduces output-polynomially to Structure Identification~\cite{Khardon1995}. Together with \cref{cor:cmi-final}, this proves \cref{cor:pseudo-output}. The correspondence between pseudo-intents and premises of minimum FD covers established by Babin and Kuznetsov~\cite{BabinKuznetsov2013} then yields \cref{cor:fd-cover}.

\section{Functional dependencies and agreement sets}
\label{app:fd}

We prove \cref{thm:fd-equivalence}. Let $R$ be a finite relation on attribute set $U$. For tuples $s,t\in R$, define their agreement set
\[
  \Agr(s,t)=\set{a\in U\suchthat s(a)=t(a)},
\]
and let
\[
  \Agr(R)=\set{\Agr(s,t)\suchthat s,t\in R}.
\]

\begin{lemma}[Agreement-set representation]
\label{lem:agreement}
The closure system determined by all functional dependencies valid in $R$ is
\[
  \C_{F_R}=\Int(\Agr(R)).
\]
\end{lemma}

\begin{proof}
By definition, a singleton-right-hand-side dependency $X\to a$ holds in $R$ exactly when every pair of tuples agreeing on $X$ also agrees on $a$, equivalently when every agreement set containing $X$ also contains $a$. Hence the closure of $X$ under $F_R$ is
\[
  \bigcap\set{A\in\Agr(R)\suchthat X\subseteq A}.
\]
The fixed points of this closure operator are precisely the intersections of agreement sets. General right-hand sides follow by decomposition.
\end{proof}

\begin{lemma}[Polynomial realization of an intersection family]
\label{lem:relation-realization}
For every family $\M\subseteq 2^U$, one can construct in polynomial time a relation $R_\M$ with at most $|\M|+1$ tuples such that
\[
  \Int(\Agr(R_\M))=\Int(\M).
\]
\end{lemma}

\begin{proof}
Create a distinguished tuple $t_0$ and one tuple $t_M$ for every $M\in\M\setminus\{U\}$. Put $t_0(a)=0$ for every attribute $a$, and define
\[
  t_M(a)=
  \begin{cases}
    0, & a\in M,\\
    \langle M,a\rangle, & a\notin M,
  \end{cases}
\]
where every symbol $\langle M,a\rangle$ is fresh and is encoded by the pair of indices of $M$ and $a$. Then
\[
  \Agr(t_0,t_M)=M
\]
and, for distinct $M,N$,
\[
  \Agr(t_M,t_N)=M\cap N.
\]
Equal tuple pairs contribute $U$. Thus every agreement set belongs to $\Int(\M)$, while every member of $\M\setminus\{U\}$ occurs as an agreement set and $U$ is present automatically. Taking intersection closures gives the claim, and the relation has polynomial encoding length.
\end{proof}

\begin{proof}[Proof of \cref{thm:fd-equivalence}]
Membership in coNP follows from a direct certificate. If $R$ fails some dependency in the supplied set $F$, two tuples witnessing that violation suffice. Otherwise, inequivalence is certified by an FD $X\to a$ that holds in $R$ but is not implied by $F$. Validity in $R$ is checked over all tuple pairs, and nonimplication is checked by verifying $a\notin\cl_F(X)$.

For hardness, take an instance $(U,\Sigma,\M)$ produced by \cref{thm:main}. Regard every implication $A\to b$ in $\Sigma$ as the singleton-right-hand-side FD $A\to b$, obtaining $F_\Sigma$. Construct $R_\M$ by \cref{lem:relation-realization}. Then \cref{lem:agreement,lem:relation-realization} gives
\[
  \C_{F_{R_\M}}=\Int(\M),
  \qquad
  \C_{F_\Sigma}=\C_\Sigma.
\]
Two FD theories are equal exactly when their closure operators, equivalently their families of closed attribute sets, are equal. Therefore
\[
  F_{R_\M}=F_\Sigma^+
  \quad\Longleftrightarrow\quad
  \Int(\M)=\C_\Sigma.
\]
This is a polynomial reduction.

Finally, the promised hard instances satisfy $\M\subseteq\C_\Sigma$, because every supplied set is a meet-irreducible of $\C_\Sigma$. Since $\C_\Sigma$ is intersection-closed,
\[
  \Int(\M)\subseteq\C_\Sigma.
\]
Inclusion of closure systems reverses inclusion of their valid implication theories, so
\[
  F_\Sigma^+\subseteq F_{R_\M}.
\]
Thus $R_\M$ satisfies every dependency in $F_\Sigma$, and equivalence can fail only through additional dependencies valid in the relation.
\end{proof}

\end{document}